\documentclass{article}

\usepackage{graphicx}
\usepackage{appendix}
\usepackage{amssymb}
\usepackage{dsfont}
\usepackage{amsmath}
\usepackage{amsthm}
\usepackage{cancel}
\usepackage{verbatim}
\usepackage{chngpage}
\usepackage{fullpage}

\usepackage[english]{babel}

\usepackage{url}

\makeatletter
\newif\if@restonecol
\makeatother

\usepackage[boxruled]{algorithm2e}

\newtheorem{thm}{Theorem}

\newtheorem{cor}[thm]{Corollary}

\newtheorem{prop}[thm]{Proposition}
\newtheorem{defn}[thm]{Definition}
\newtheorem{exmp}[thm]{Example}

\begin{document}

\title{Feasibility of Motion Planning\\ on Acyclic and Strongly Connected Directed Graphs}
\author{Zhilin Wu and St\'{e}phane Grumbach \\
INRIA-LIAMA\thanks{CASIA -- PO Box 2728 --
Beijing 100080 -- PR China -- Stephane.Grumbach@inria.fr  zlwu@liama.ia.ac.cn}\\
Chinese Academy of Sciences}

\maketitle

\begin{abstract}
Motion planning is a fundamental problem of robotics with
applications in many areas of computer science and beyond. Its
restriction to graphs has been investigated in the literature for it
allows to concentrate on the combinatorial problem abstracting from
geometric considerations.  In this paper, we consider motion
planning over directed graphs, which are of interest for asymmetric
communication networks. Directed graphs generalize undirected
graphs, while introducing a new source of complexity to the motion
planning problem: moves are not reversible. We first consider the
class of acyclic directed graphs and show that the feasibility can
be solved in time linear in the product of the number of vertices
and the number of arcs. We then turn to strongly connected directed
graphs. We first prove a structural theorem for decomposing strongly
connected directed graphs into strongly biconnected components.
Based on the structural decomposition, we give an algorithm for the
feasibility of motion planning on strongly connected directed
graphs, and show that it can also be decided in time linear in the
product of the number of vertices and the number of arcs.
\end{abstract}

\section{Introduction}

Motion planning is a fundamental problem of robotics. It has been
extensively studied \cite{Lavalle06}, and has numerous practical
applications beyond robotics, such as in manufacturing, animation,
games \cite{MPG} as well as in computational biology
\cite{SA01,FK99}. The complexity of motion planning, which is
intrinsically PSPACE-hard \cite{Latombe95,Lavalle06}, has received a
lot of attention. The study of motion planning on graphs was
proposed by Papadimitriou et al. \cite{PapadimitriouRST94} to strip
away the geometric considerations of the general motion planning
problem and concentrate on the combinatorial problem.

In this paper, we consider the feasibility of motion planning over
directed graphs. Our results generalize results on undirected
graphs, which can be shown as a subclass of directed graphs.
Directed graphs are of great importance in several fields such as
communication networks which are frequently asymmetric
\cite{JetchevaJ06}. But technically, directed graphs differ from
undirected graphs, for movements in the graph are not reversible.

Papadimitriou et al. \cite{PapadimitriouRST94} first introduced the
problem of motion planning on graphs. They defined the Graph Motion
Planning with 1 Robot problem (GMP1R) as follows: Suppose we are
given a graph $G=(V,E)$ with $n$ vertices, there is one robot in a
vertex $s$ and some of the other vertices contain a movable
obstacle. The objective of GMP1R is to move the robot from the
source vertex $s$ to a destination vertex $t$ with the smallest
number of moves, where a move consists in moving a robot or an
obstacle from one vertex to an adjacent vertex that does not contain
an object (robot or obstacle). It may be impossible to move the
robot from $s$ to $t$, for instance, if all the vertices other than
$s$ are occupied by obstacles. The feasibility problem of GMP1R is
to decide whether it is possible or not to move the robot from the
source vertex to the destination vertex.

In \cite{PapadimitriouRST94}, it was shown that the feasibility of
GMP1R can be decided in polynomial time, and the optimization of
GMP1R is NP-complete (even on planar graphs). They also gave a
$O(n^6)$ exact algorithm as well as a fast $7$-approximation
algorithm for GMP1R on trees, a $O(\sqrt{n})$-approximation
algorithm for GMP1R on general graphs. Auletta et al. proposed more
efficient algorithms for the feasibility and optimization of GMP1R
on trees in \cite{AulettaMPP96,AulettaP01}.

Motion planning on graphs has wide practical applications. Track
transportation system \cite{Per88} constitutes a typical example:
Vehicles move on a system of tracks such that each track connects
two distinct stations. There is a distinguished vehicle which moves
from a source station to a destination station. There are other
vehicles (obstacles) on the non-source stations. The vehicles are
only able to stop at the stations and not able to stop in the middle
of tracks. They coordinate with each other to let the distinguished
vehicle move from the source station to the destination station.
Variant of the previous example is packet transfer in communication
buffers. Graphs are regarded as (bidirectional) communication
networks and objects as indivisible packets of data. If there is a
distinguished packet which moves from a source node to a destination
node, and there are already some other packets stored in the
communication buffers of nodes in the network, the objective is to
move the distinguished packet from the source node to the
destination node without exceeding the capacities of the
communication buffers of each node.



In practice, in the two previous examples, the tracks (links)
between the stations (nodes) might be asymmetric. This motivates the
study of motion planning on directed versus undirected graphs.

Let us consider first the track transportation system. Some of the
tracks may be unidirectional. For instance, if the two stations are
not in the same altitude, the track connecting them might be too
steep, and the vehicles not strong enough to climb up the track.
There may also be unidirectional tracks as a result of security
considerations. The vehicle movement from a source station to a
destination station, on a track-transportation system containing
unidirectional tracks, leads to motion planning on directed graphs.

Now consider the packet transfer in communication networks. There
might be unidirectional links in communication networks. For
instance, in wireless ad hoc networks, unidirectional links can
result from factors such as heterogeneity of receiver and
transmitter hardware, power control algorithms, or topology control
algorithms. Unidirectional links may also result from interferences
around a node that prevents it from receiving packets even though
the other nodes are able to receive packets from it
\cite{MD02,JetchevaJ06}. Networks with unidirectional links can be
modeled as directed graphs. The problem of transferring a
distinguished packet in networks with unidirectional links without
exceeding the capacities of the communication buffers amounts to
solving motion planning on directed graphs.

Directed graphs generalize undirected graphs, while introducing a
new source of complexity to the motion planning problem: moves are
not reversible, and motion planning might become infeasible after
inappropriate moves.

In this paper, we first give a motivating example to illustrate that
motion planning on directed graphs is much more intricate than
motion planning on graphs. Then, we consider the class of acyclic
directed graphs, we give an algorithm to decide the feasibility on
such class of directed graphs, prove its correctness, and analyze
its complexity. We show that the feasibility of motion planning on
acyclic directed graphs can be decided in time linear in the product
of the number of vertices and the number of arcs
(Theorem~\ref{thm:acyclic-complexity}). We then turn to strongly
connected directed graphs. We first consider their structure and
introduce a new class of directed graphs, strongly biconnected
directed graphs. We obtain an interesting characterization of
strongly biconnected directed graphs by showing that a directed
graph is strongly biconnected iff it has an open ear decomposition
(Theorem~\ref{thm:bi-strong}). This characterization can be seen as
a generalization of the classical open-ear-decomposition
characterization of biconnected graphs. We then prove a structural
theorem for decomposing strongly connected directed graphs into
strongly biconnected components (Theorem~\ref{thm:strong2}). Based
on the open-ear-decomposition characterization, we show that motion
planning on strongly biconnected directed graphs is feasible iff
there is at least one vertex occupied neither by robot nor by
obstacle (Theorem~\ref{thm:feas-bi-strong}). Based on the structural
decomposition, we give an algorithm for the feasibility of motion
planning on strongly connected directed graphs, prove its
correctness, and analyze its complexity. We show that the
feasibility of motion planning on strongly connected directed graphs
can also be decided in time linear in the product of the number of
vertices and the number of arcs
(Theorem~\ref{thm:strong-complexity}).

The paper is organized as follows. A motivating example is presented
in the next section. In Section~\ref{sec:prelim}, we recall
classical definitions from graph theory. We consider acyclic
directed graphs in Section~\ref{sec:acyclic}, and give an algorithm
to decide the feasibility of motion planning on such class of
directed graphs. In Section~\ref{sec:strong-structure}, we consider
strongly biconnected directed graphs, and prove a structural theorem
on their decomposition into strongly biconnected components. The
feasibility for strongly connected directed graphs is considered in
Section~\ref{sec:strong-algo}.

\section{A motivating example}
In the sequel, for brevity, we use ``digraph'' to denote ``directed
graph''.

Let us consider a simple example to illustrate the motion planning
on digraphs. Vertices can contain either an object (obstacle or the
robot) or nothing. If there is no object on a vertex, we say that
there is a \emph{hole} in that vertex. For an arc $(v,w)$ from $v$
to $w$, with an object on $v$, and a hole on $w$, the object can be
moved from $v$ to $w$, and we say equivalently that the hole can be
moved (backwards) from $w$ to $v$.

\begin{figure}[ht]
\centering
\includegraphics[width=0.7\textwidth]{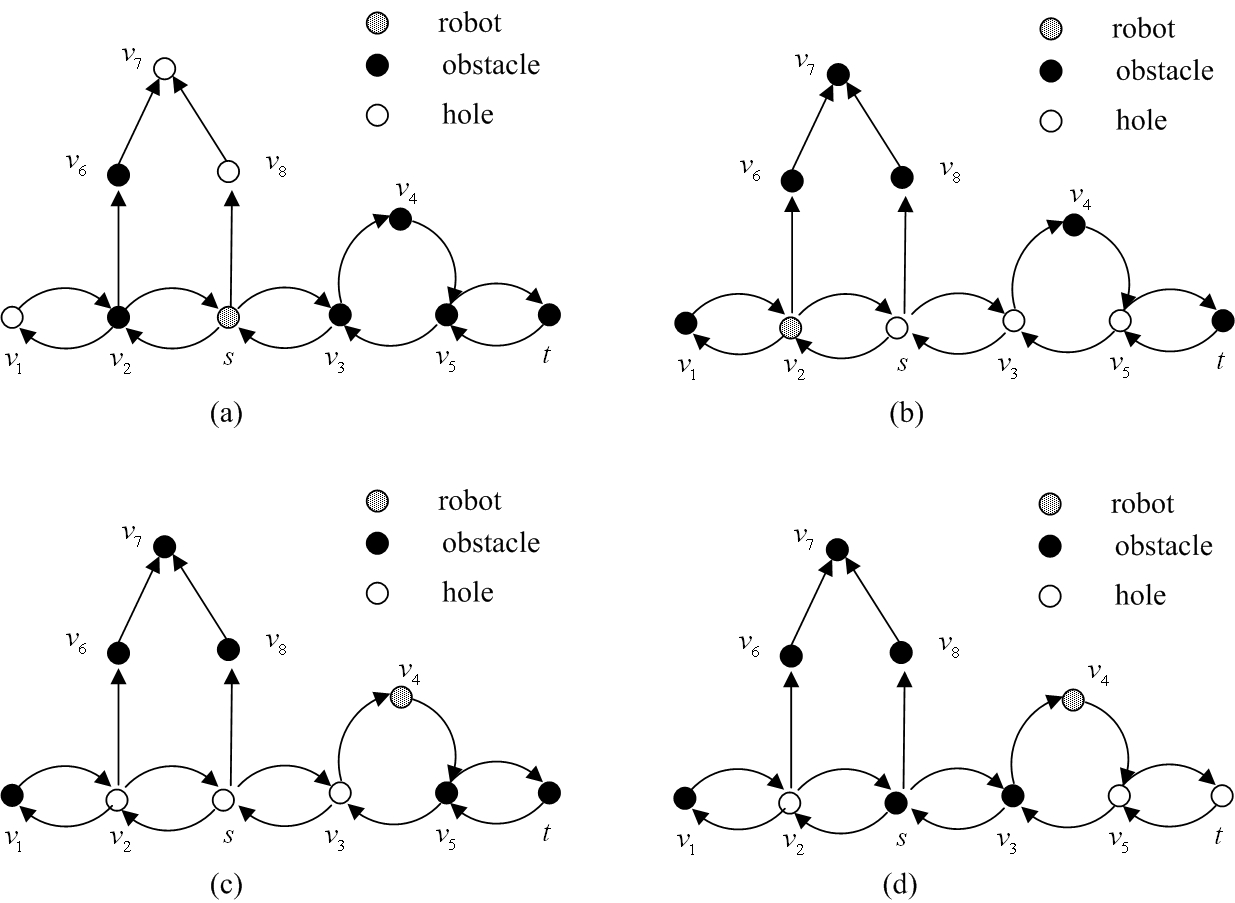}
\caption{Motion planning on digraphs}\label{fig:example-1}
\end{figure}

Consider the strongly connected component  $C$ in the graph of
Figure~\ref{fig:example-1} which contains vertices $v_1, ..., v_5$,
$s$ and $t$. The initial positions of the robot and obstacles are
shown in Figure~\ref{fig:example-1}(a).

We can move the robot from $s$ to $t$ as follows: move the hole in
$v_1$ to $v_2$, move the robot from $s$ to $v_2$, then move the two
holes in $v_{7}$, $v_{8}$ into $C$ through $s$, without moving the
robot in $v_2$ (Figure~\ref{fig:example-1}(b)). Now move the
obstacle in $v_4$ to $v_5$, and move the robot to $v_4$ (see
Figure~\ref{fig:example-1}(c)). Move the two obstacles in $v_5$ and
$t$ to $v_3$ and $s$ (Figure~\ref{fig:example-1}(d)), then move the
robot from $v_4$ to $v_5$, and finally to $t$. The main idea of
these moves is to move the robot to $v_4$ in order to free the way
for the moves of the holes from $s$ and $v_3$  to $v_5$ and $t$.

If the robot is in $s$ and we move the hole in $v_7$ to $v_2$
(Figure~\ref{fig:example-2}(a)), then the problem becomes
infeasible. We can move the robot from $s$ to $v_2$ and the hole in
$v_8$ to $s$ (Figure~\ref{fig:example-2}(b)), but it is then
impossible to move the robot from $v_2$ to $v_4$.

\begin{figure}[ht]
\centering
\includegraphics[width=0.7\textwidth]{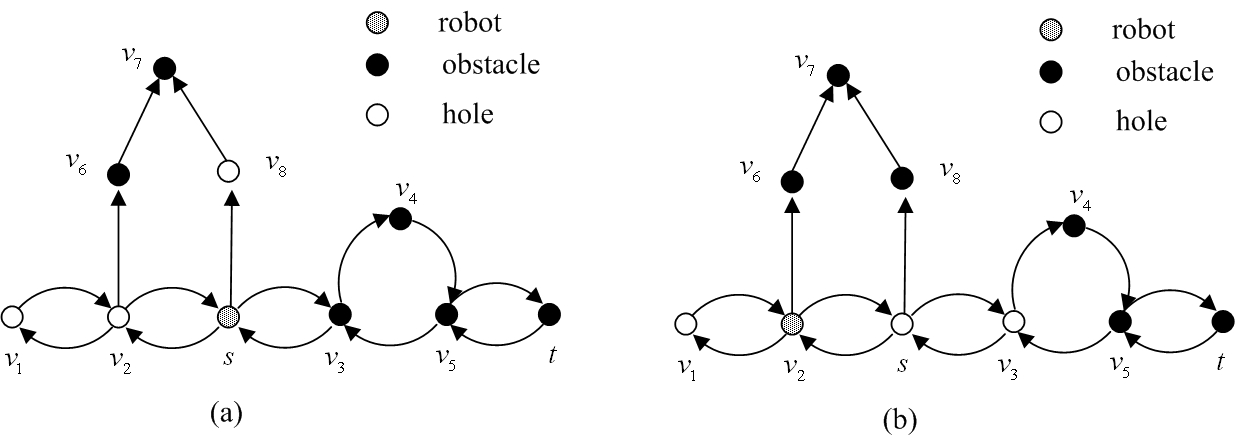}
\caption{Object moves that do not preserve feasibility}\label{fig:example-2}
\end{figure}

As illustrated in the above example, the intricacy of motion
planning on digraphs follows from the non-reversibility of moves in
the digraphs.

In the sequel, we propose algorithms which take as input, a digraph
$D=(V,E)$ encoded by its adjacency lists, a source and destination
vertex $s,t \in V$, a function $f$ mapping each vertex to an element
of the set \{``robot'',``obstacle'', ``hole''\}, and produces a
Boolean value (true or false) indicating whether it is feasible to
move the robot from $s$ to $t$ in $D$.

\section{Preliminaries}\label{sec:prelim}

A \emph{digraph} $D$ is a binary tuple $(V, E)$ such that $E
\subseteq V^2$. Elements of $V$ and $E$ are called respectively
\emph{vertices} and \emph{arcs} of $D$. We assume that $(v,v) \not
\in E$ for all $v \in V$ (there are no self-loops).

For a vertex $v$ of a digraph $D=(V,E)$, the \emph{indegree} of $v$,
denoted $in(v)$, is defined as $\lvert\{w \in V | (w,v) \in
E\}\rvert$, and the \emph{outdegree} of $v$, denoted $out(v)$, is
defined as $|\{w \in V | (v,w) \in E\}|$.

A \emph{graph} $G$ is a binary tuple $(V, E)$ such that $E \subseteq
V^{[2]}$, where $V^{[2]}$ contains exactly all two-element subsets
of $V$, namely $V^{[2]}=\{\{v,w\}| v, w \in V, v \neq w\}$. Elements
of $E$ are called \emph{edges} of $G$.

For a vertex $v$ of a graph $G=(V,E)$, the \emph{degree} of $v$,
denoted $deg(v)$, is defined as $|\{w \in V | \{v,w\} \in E\}|$.

If $D=(V,E)$ (resp. $G=(V,E)$), and $e=(v,w) \in E$ (resp.
$e=\{v,w\} \in E$), then $e$ is said to be \emph{incident} to $v$
and $w$ in $D$ (resp. $G$).

A digraph (resp. graph) containing exactly one vertex is said to be
\emph{trivial}, otherwise it is said to be \emph{nontrivial}.
%
%

Given a digraph $D=(V,E)$ (resp. graph $G=(V,E)$), the digraph
(resp. graph) $H=(V_H,E_H)$ such that $V_H \subseteq V$ and $E_H
\subseteq E$ is called a \emph{sub-digraph} of $D$ (resp.
\emph{subgraph} of $G$). Let $X \subseteq V$, the sub-digraph (resp.
subgraph) \emph{induced by} $X$, denoted $D[X]$ (resp. $G[X]$), is
the sub-digraph (resp. subgraph) $(X,E \cap X \times X)$ (resp. $(X,
E \cap X^{[2]})$).


Suppose $D=(V,E)$ (resp. $G=(V,E)$) is a digraph (resp. graph) and
$X \subseteq V$, let $D-X$ (resp. $G-X$) denote the digraph (resp.
graph) obtained from $D$ (resp. $G$) by deleting all the vertices in
$X$ and all the arcs (resp. edges) incident to at least one element
of $X$. If $X=\{v\}$, then $D-\{v\}$ (resp. $G-\{v\}$) is written as
$D-v$ (resp. $G-v$) for simplicity.

Given a digraph $D=(V,E)$, the \emph{underlying graph} of $D$,
denoted by $\mathcal{G}(D)$, is the graph obtained from $D$ by
omitting the directions of arcs, namely $\mathcal{G}(D)=\left(V,
\{\{v,w\} | (v,w) \in E\}\right)$.

A \emph{path} of a digraph $D=(V,E)$ (resp. graph $G=(V,E)$) is an
alternating sequence of vertices and arcs (resp. edges)
$v_0e_1v_1...v_{k-1}e_kv_k$ ($k \ge 1$) such that for all $1 \le i
\le k$, $e_i=(v_{i-1}, v_i) \in E$ (resp. $e_i=\{v_{i-1},v_i\} \in
E$), and for all $0 \le i < j \le k$, $v_i \ne v_j$. $v_0$ and $v_k$
are called the \emph{tail} and \emph{head endpoint} of the path
respectively, and the other vertices are called the \emph{internal
vertices} of the path. In particular, an arc or an edge is a path
without internal vertices.

A \emph{cycle} of a digraph $D=(V,E)$ is a sequence of vertices
$v_0v_1...v_k$ such that for all $0 \le i \le k$, $(v_i, v_{i+1})
\in E$ ($v_{k+1}$ interpreted as $v_0$), and for all $0 \le i < j
\le k$, $v_i \ne v_j$. Cycles of graphs can be defined similarly,
but we have the additional restriction that $k \ge 2$. So cycles of
graphs contain at least three vertices.

A digraph $D$ is \emph{acyclic} if there are no cycles in $D$.

Suppose $H=(V_H,E_H)$ is a sub-digraph of $D=(V, E)$ (resp. subgraph
of $G=(V,E)$). A path $P$ of $D$ (resp. $G$) is an \emph{$H$-path}
if the two endpoints of $P$ are in $H$, no internal vertices of $P$
are in $H$, and no arcs (resp. edges) of $P$ are in $H$. In
particular, an arc $(v,w) \in E \backslash E_H$ (resp. an edge
$\{v,w\} \in E \backslash E_H$) with $v,w \in V_H$ is an $H$-path. A
cycle $C$ is an \emph{$H$-cycle} if there is exactly one vertex of
$C$ in $H$.

Let $H_1=(V_1,E_1)$ and $H_2=(V_2,E_2)$ be two sub-digraphs of a
digraph $D=(V,E)$, then the \emph{union} of $H_1$ and $H_2$, $H_1
\cup H_2$, is defined as $\left(V_1 \cup V_2, E_1 \cup E_2\right)$.
The union of subgraphs can be defined similarly.

A digraph $D=(V,E)$ is \emph{strongly connected} if for any two
distinct vertices $v$ and $w$, there are both a path from $v$ to $w$
and a path from $w$ to $v$ in $D$. The digraph containing exactly
one vertex and no arcs is the minimal strongly connected digraph.

Let $D=(V,E)$ be a digraph. The \emph{strongly connected components}
of $D$ are the maximal strongly connected sub-digraphs of $D$.

A graph $G=(V,E)$ is \emph{connected} if for any two distinct
vertices $v$ and $w$ of $G$, there is a path of $G$ with endpoint
$v$ and $w$. The \emph{connected components} of a graph $G$ are the
maximal connected subgraphs of $G$.


If $G=(V,E)$ is a graph, $v \in V$, and the number of connected
components of $G-v$ is more than that of $G$, then $v$ is said to be
a \emph{cut vertex} of $G$.

A graph $G$ is \emph{biconnected} if $G$ is connected and there are
no cut vertices in $G$. In particular, the graph containing exactly
one vertex is the minimal biconnected graph. The \emph{biconnected
components} of a graph $G$ are the maximal biconnected subgraphs of
$G$.

\medskip \noindent Without loss of generality, we assume
that for each digraph $D$, (i) the underlying graph of $D$,
$\mathcal{G}(D)$, is connected, (ii) the source vertex $s$ and the
destination vertex $t$ are distinct (thus all the digraphs
considered from now on are nontrivial), (iii) there is at least one
path from $s$ to $t$ in $D$.


\section{Motion planning on acyclic digraphs}\label{sec:acyclic}
In this section we assume that $D=(V,E)$ is an acyclic digraph.

We first recall a result about acyclic orderings of acyclic
digraphs.

An \textit{acyclic ordering} of an acyclic digraph $D=(V,E)$ is an
ordering of all vertices of $D$, say $v_1,\dots,v_k$, such that
$(v_i,v_j) \in E$ implies $i<j$. From \cite{BangGutin00}, we know
that an acyclic ordering of a given acyclic digraph can be computed
in linear time by depth-first-search.

\begin{thm}[\cite{BangGutin00}]\label{thm:acyclic-ordering}
Given an acyclic digraph $D=(V,E)$, an acyclic ordering of $D$ can
be computed in time $O(n+m)$, where $n$ is the number of vertices
and $m$ is the number of arcs of $D$.
\end{thm}

We introduce some notations in the following.

Let $V^\prime$ denote the set of vertices from which there is a path
to $t$, and to which there is a path from $s$. In particular, $s,t
\in V^\prime$.

For each $v \in V^\prime$, let $h(v)$ denote the number of holes
that can be moved to $v$.

For each $v \in V^\prime$, define $h_t(v)$ as follows: Suppose that
the robot is in $v$.
\begin{itemize}
\item If the robot can be moved from $v$ to $t$ in $D$, then there
may be different paths (from $v$ to $t$) along which the robot can
be moved from $v$ to $t$, let $h_t(v)$ be the minimal length (number
of arcs) of such paths.

\item If it is impossible to move the robot from $v$ to $t$, let
$h_t(v)= \infty$.
\end{itemize}

\medskip

The algorithm FAD($D,s,t,f$) (see Algorithm 1 in the box below)
decides the feasibility of the motion planning problem on acyclic
digraphs. FAD first computes $h(v)$ for each $v \in V^\prime$, then
computes $h_t(v)$ for each $v \in V^\prime$, finally checks whether
$h_t(s)<\infty$.

%
%
%
%
%
\begin{algorithm}[ht]

\dontprintsemicolon \SetVline
%
\SetKwInOut{Input}{input}
\SetKwInOut{Output}{output}
\Input{$(D,s,t,f)$ such that $D=(V,E)$ is an acyclic digraph, $s,t \in V$, $s \ne t$, and $f$ is a function from $V$ to \{``robot'',``obstacle'',``hole''\}.}
\Output{true or false.}
\tcc{Compute $V^\prime$, the set of vertices reachable from $s$ and
from which $t$ is reachable.}
Let $W$ be a FIFO queue, push $s$ into $W$.\;
\lForEach{$v \in V$} {\lIf{$v =s$}{$srcReach(v):=true$.
}\lElse{$srcReach(v):=false$.}}\;
\While{$W$ is nonempty}
{
  $w:=$the first element of $W$, pop the first element of $W$.\;
  \ForEach{$w^\prime$ such that $(w,w^\prime) \in E$}
  {
    \If{$srcReach(w^\prime)=false$}
    {$srcReach(w^\prime):=true$, push $w^\prime$ into $W$.}
  }
}
%
%
Push $t$ into $W$.\;
\lForEach{$v \in V$}{\lIf{$v = t$}{$reachDest(v):=true$
}\lElse{$reachDest(v):=false$}}\;
\While{$W$ is nonempty}
{
  $w:=$the first element of $W$, pop the first element of $W$.\;
  \ForEach{$w^\prime$ such that $(w^\prime,w) \in E$}
  {
    \If{$reachDest(w^\prime)=false$}
    {$reachDest(w^\prime):=true$, push $w^\prime$ into $W$.}
  }
}
$V^\prime:=\{v \in V | srcReach(v)=reachDest(v)=true\}$.\;
\tcc{Compute $h(v)$}
\ForEach{$v \in V^\prime$}
{
  \lForEach{$v^\prime \in V$}{$idx(v^\prime):=false$}\;
%
%
  \lIf{$f(v)=$``hole''}{$h(v):=1$. }\lElse{$h(v):=0$.}\;
  Push $v$ into $W$.\;
  \While{$W$ is nonempty}
  {
    $w:=$ the first element of $W$, pop the first element of $W$.\;
    \ForEach{$w^\prime \in V$ such that $(w,w^\prime) \in E$}
    {
      \If{$idx(w^\prime)=false$}
      {
        $idx(w^\prime):=true$, push $w^\prime$ into $W$.\;
        \lIf{$f(w^\prime)=$``hole''}{$h(v):=h(v)+1$.}
      }
    }
  }
}
\tcc{Compute $h_t(v)$}
Compute an acyclic ordering of $D[V^\prime]$, say $v_1,\dots,v_k$,
such that $v_1=s,v_k=t$.\;
$h_t(v_k):=0$.\;
\For{$i$ from $k-1$ to $1$}
{
    \If{$\exists j: i < j \le k$ such that
    $(v_i,v_j) \in E$ and $h(v_j) \ge
    h_t(v_j)+1$}
      {
        $h_t(v_i):=\min\{h_t(v_j)+1 |
        (v_i,v_j) \in E, h(v_j) \ge
        h_t(v_j)+1 \}$.\;
      }
}
\lIf{$h_t(s) < \infty$}{\Return true. }\lElse{\Return false.}\;
\caption{FAD($D,s,t,f$)\label{Alg:feas-acyclic}}
\end{algorithm}

\clearpage

The computation of $h(v)$ and $h_t(v)$ on an acyclic digraph is
illustrated in Figure~\ref{fig:acyclic-example}.

\begin{figure}[ht]
\centering
  \includegraphics[width=0.7\textwidth]{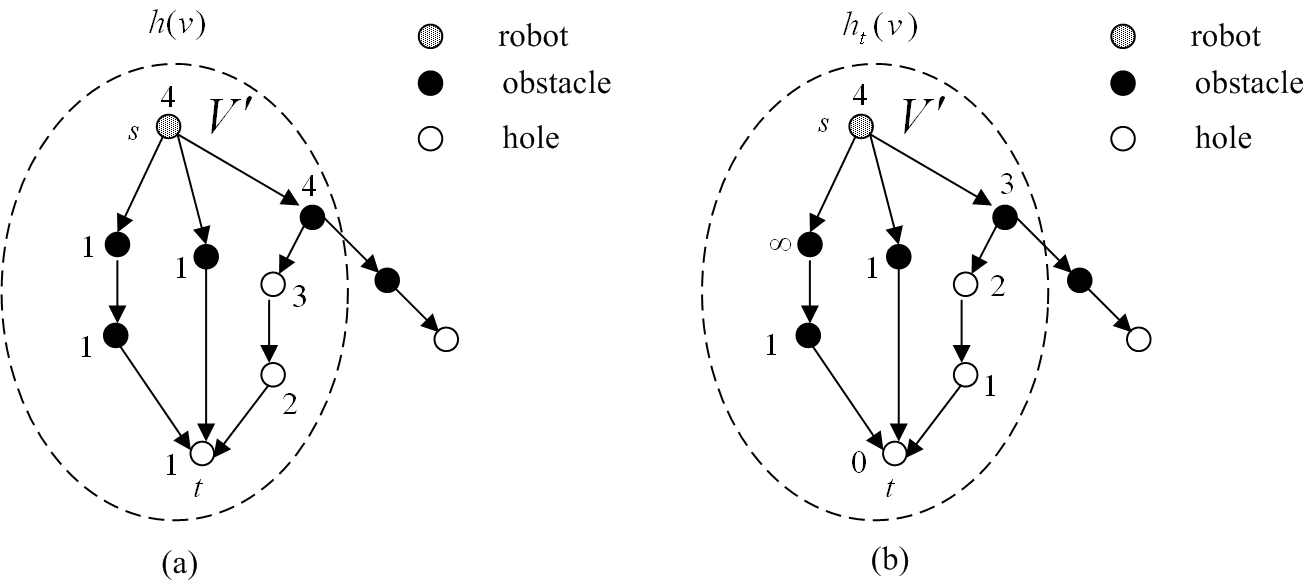}
  \caption{Computation of $h(v)$ and $h_t(v)$}\label{fig:acyclic-example}
\end{figure}


\begin{thm}
FAD is correct.
\end{thm}


\begin{proof}

We prove that given an instance of the motion planning problem on
acyclic digraphs, FAD returns true iff the problem is feasible.

\noindent ``If'' part:

If the problem is feasible, the robot can be moved from $s$ to $t$.
Let $P$ be the trace of the robot during this movement (namely the
sequence of nodes and arcs reached by the robot). As a result of
acyclicity of $D$, $P$ is a path of $D$. Let $P=v_0e_1v_1\cdots
v_{k-1}e_k v_k$ such that $v_0=s$ and $v_k=t$. During the movement,
when the robot is moved to $v_{i-1}$($1 \le i \le k$), in order to
move the robot from $v_{i-1}$ to $v_{i}$, a hole should be moved to
$v_{i}$. Since $D$ is acyclic, the hole moved to $v_{i}$ cannot be
moved to $v_j$ for any $j$ such that $i < j \le k$ (holes are moved
along the reverse direction of arcs). So these holes are distinct
from each other, and can be moved to occupy all the vertices on $P$
except $s$. By induction, we can show that for all vertices $v$ on
$P$, $h_t(v)$ computed by FAD satisfies that $h_t(v) < \infty$.
Consequently FAD returns true.

\noindent ``Only if'' part:

If FAD returns true, then $h_t(s) < \infty$. By induction, we can
show that there is a path $P$ from $s$ to $t$ such that for each
vertex $v \ne s$ on $P$, we have $h_t(v)< \infty$ and $h(v) \ge
h_t(v)+1$, and for each arc $(v,w)$ on $P$, $h_t(v)=h_t(w)+1$. By
induction again, we can show that the holes in $D$ can be moved to
occupy all the vertices on $P$ except $s$. Then the robot can be
moved to $t$ along $P$, the problem is feasible.
\end{proof}

\begin{thm}\label{thm:acyclic-complexity}
The time complexity of FAD is $O(nm)$, where $n$ is the number of
vertices, and $m$ is the number of arcs.
\end{thm}

\begin{proof}

Let $D$ be an acyclic digraph, $n$ and $m$ be the number of vertices
and number of arcs of $D$ respectively.

The computation of $V^\prime$ takes $O(m)$ time since each arc is
visited at most once in each of the first two ``While'' loops.

The computation of $h(v)$'s takes $O(nm)$ time because the
computation of each $h(v)$ takes $O(m)$ time and there are at most
$O(n)$ such computations.

The computation of an acyclic ordering of $D[V^\prime]$ takes
$O(n+m)$ time from Theorem~\ref{thm:acyclic-ordering}.

The computation of $h_t(v)$'s takes $O\left(\sum \limits_{v \in
V^\prime}out(v)\right)=O(m)$ time.

Since $n \le m$, we conclude that the time complexity of FAD is
$O(m+nm+n+m+m)=O(nm)$.
\end{proof}

\section{Structure of strongly connected digraphs}\label{sec:strong-structure}

In this section, we consider the structure of strongly connected
digraphs. We first recall some definitions and theorems.

An \emph{open ear decomposition} of a digraph $D=(V,E)$ (resp. graph
$G=(V,E)$) is a sequence of sub-digraphs of $D$ (resp. subgraphs of
$G$), say $P_0, ..., P_r$, such that
\begin{itemize}

\item $P_0$ is a cycle;

\item $P_{i+1}$ is a $D_i$-path (resp. $G_i$-path), where $D_i$ (resp. $G_i$) is $\bigcup \limits_{0 \le j \le i} P_j$ for all $0 \le i <
r$;

\item $D = \bigcup \limits_{0 \le i \le r}P_i$ (resp. $G = \bigcup \limits_{0 \le i \le r}P_i$).
\end{itemize}

A \emph{closed ear decomposition} of a digraph $D=(V,E)$ (resp.
graph $G=(V,E)$) is a sequence of sub-digraphs of $D$ (resp.
subgraphs of $G$), say $P_0, ..., P_r$, such that
\begin{itemize}

\item $P_0$ is a cycle;

\item $P_{i+1}$ is a $D_i$-path or a $D_i$-cycle (resp. a $G_i$-path or a $G_i$-cycle), where $D_i$ (resp. $G_i$) is $\bigcup \limits_{0 \le j \le i} P_j$ for all $0 \le i <
r$;

\item $D = \bigcup \limits_{0 \le i \le r}P_i$ (resp. $G = \bigcup \limits_{0 \le i \le r}P_i$).
\end{itemize}

\begin{thm}[\cite{West00}]\label{thm:bi-connected}
Let $G$ be a graph containing at least three vertices. $G$ is
biconnected iff $G$ has an open ear decomposition. Moreover, any
cycle can be the starting point of an open ear decomposition.
\end{thm}

\begin{thm}[\cite{BangGutin00}]\label{thm:strong}
Let $D$ be a nontrivial digraph. $D$ is strongly connected iff $D$
has a closed ear decomposition. Moreover, any cycle can be the
starting point of a closed ear decomposition.
\end{thm}

Let $G=(V,E)$ be a graph. The \emph{biconnected-component graph} of
$G$, denoted $\mathcal{G}_{bc}(G)$, is a bipartite graph
$(V_{bc},W_{bc},E_{bc})$ defined by
\begin{itemize}
\item $V_{bc}$: biconnected components of $G$;

\item $W_{bc}$: vertices of $G$ shared by
at least two distinct biconnected components of $G$;

\item $E_{bc}$: let $B \in V_{bc}$ and $w \in W_{bc}$, then $(B,w) \in E_{bc}$
iff $w \in V(B)$.
\end{itemize}

\begin{thm}[\cite{West00}]\label{thm:biconn-components}
Let $G=(V,E)$ be a connected graph. Then $\mathcal{G}_{bc}(G)$ is a
tree.
\end{thm}

\smallskip

\noindent Now we introduce a new class of digraphs, strongly
biconnected digraphs.

\begin{defn}
Let $D$ be a digraph. $D$ is said to be strongly biconnected if $D$
is strongly connected and $\mathcal{G}(D)$ is biconnected. The
strongly biconnected components of $D$ are the maximal strongly
biconnected sub-digraphs of $D$.
\end{defn}
%
%
%
In particular, the digraph containing exactly one vertex and no arcs
is strongly biconnected.

\noindent We now show that strongly biconnected digraphs also admit
a similar characterization.

\begin{thm}\label{thm:bi-strong}
Let $D$ be a nontrivial digraph. $D$ is strongly biconnected iff $D$
has an open ear decomposition. Moreover, any cycle can be the
starting point of an open ear decomposition.
\end{thm}
\begin{proof}

\noindent \textbf{``If'' part}: Suppose $D$ has an open ear
decomposition $P_0,\dots,P_r$.

Since open ear decompositions are special cases of closed ear
decompositions, from Theorem~\ref{thm:strong}, we know that $D$ is
strongly connected.

Let $P^\prime_i=\mathcal{G}(P_i)$, the underlying graph of $P_i$,
for all $0 \le i \le r$.

If $P_0$ is a cycle with at least $3$ vertices, then
$P^\prime_0,...,P^\prime_r$ is an open ear decomposition of
$\mathcal{G}(D)$, $\mathcal{G}(D)$ is biconnected according to
Theorem~\ref{thm:bi-connected}. So $D$ is strongly biconnected.

If $P_0$ is a cycle with only two vertices and $r=0$, then
$\mathcal{G}(D)$ is a graph with exactly two vertices connected by
an edge. $\mathcal{G}(D)$ is biconnected and $D$ is strongly
biconnected.

Otherwise, $P_0$ is a cycle with only two vertices and $r > 0$. Then
it is easy to see that $P^\prime_0 \cup P^\prime_1$ is a cycle of
$\mathcal{G}(D)$, so $\left(P^\prime_0 \cup
P^\prime_1\right),P^\prime_2,...,P^\prime_r$ is an open ear
decomposition of $\mathcal{G}(D)$. $\mathcal{G}(D)$ is biconnected
from Theorem~\ref{thm:bi-connected}. $D$ is strongly biconnected.

\noindent \textbf{``Only if'' part}: Suppose $D$ is nontrivial and
strongly biconnected. Consider the following procedure:

Initially select an arbitrary cycle in $D$, let $P_0$ be this cycle.

Suppose we have obtained $D_i= \bigcup \limits_{0 \le j \le i} P_j$.

Select a $D_i$-path in $D$ as $P_{i+1}$.

Continue until $D_i=D$.

The above procedure produces the desired open ear decomposition of
$D$, which is guaranteed by the following claim.

\smallskip
\noindent \textbf{Claim}. If $D_i \neq D$, there must be a
$D_i$-path in $D$.
\smallskip

\noindent \emph{Proof of the Claim}.

If $V(D)=V(D_i)$, then there must be arcs in $D$ but not in $D_i$,
which are the $D_i$-paths in $D$.

Otherwise, $V(D) \backslash V(D_i)$ is nonempty.

To the contrary, suppose that there are no $D_i$-paths in $D$.

For all $v \in V(D) \backslash V(D_i)$, we call the path from some
vertex in $D_i$ to $v$ such that none of its internal vertices are
in $D_i$, as the $(D_i,v)$-path, and the path from $v$ to some
vertex in $D_i$ such that none of its internal vertices are in
$D_i$, as the $(v, D_i)$-path. Moreover, we call the endpoint of a
$(D_i,v)$-path (resp.$(v,D_i)$-path ) that is in $D_i$ as the
$D_i$-endpoint of the $(D_i,v)$-path (resp. $(v,D_i)$-path).

Because $D$ is strongly connected, for all $v \in V(D) \backslash
V(D_i)$, there are $(D_i,v)$-paths and $(v,D_i)$-paths in $D$. Let
$\textrm{End}(D_i,v)$ and $\textrm{End}(v,D_i)$ be the set of
$D_i$-endpoints of $(D_i,v)$-paths and $(v,D_i)$-paths respectively.

For each $v \in V(D) \backslash V(D_i)$, if $v^{\prime} \in
\textrm{End}(D_i,v)$ and $v^{\prime\prime} \in \textrm{End}(v,D_i)$,
then we must have $v^\prime = v^{\prime\prime}$, because otherwise
we will have a $D_i$-path in $D$, contradicting to the assumption.
Therefore for each $v \in V(D) \backslash V(D_i)$,
$\textrm{End}(D_i,v)=\textrm{End}(v,D_i)$, and $\textrm{End}(D_i,v)$
is a singleton.

For all $v^\prime \in V(D_i)$, let $F_{v^{\prime}}$ be the set of
all $v \in V(D) \backslash V(D_i)$ such that
$\textrm{End}(D_i,v)=\{v^\prime\}$.

Then all the nonempty $F_{v^{\prime}}$'s ($v^{\prime} \in V(D_i)$)
form a partition of $V(D) \backslash V(D_i)$ (see
Figure~\ref{fig:bi-strong}).

\begin{figure}[ht]
\centering
  \includegraphics[width=0.25\textwidth]{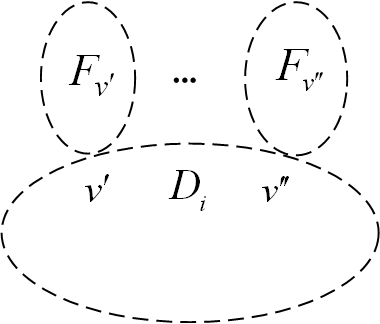}
  \caption{Partition of $V(D) \backslash V(D_i)$}\label{fig:bi-strong}
\end{figure}

Select an arbitrary $v^\prime \in V(D_i)$ such that $F_{v^{\prime}}$
is nonempty. We show that $v^\prime$ is a cut vertex in
$\mathcal{G}(D)$.

Let $v$ be a vertex in $F_{v^{\prime}}$, we show that all the arcs
incident to $v$ are confined to $F_{v^{\prime}} \cup \{v^\prime\}$,
namely if $e=(v,w) \in E$ or $e=(w,v) \in E$, then $w\in
F_{v^{\prime}}\cup \{v^\prime\}$.

To the contrary, suppose that there is an arc $e=(v,w)$ (the case of
$e=(w,v)$ is similar) such that $w \not \in F_{v^{\prime}} \cup
\{v^\prime\}$, then either $w \in V(D_i)$ and $w \neq v^\prime$, or
$w \in F_{v^{\prime\prime}}$ for some $v^{\prime\prime} \in V(D_i)$
such that $v^{\prime\prime} \neq v^\prime$.

Since $\textrm{End}(D_i,v)=\textrm{End}(v,D_i)$ is a singleton, the
former case is impossible.

For the latter case, we can get a $D_i$-path in $D$, contradicting
to the assumption.

Therefore, all the arcs incident to $v$ are confined to
$F_{v^{\prime}} \cup \{v^\prime\}$, as a consequence, all the paths
in $\mathcal{G}(D)$ from $v$ to vertices in $\mathcal{G}(D_i)$ go
through $v^\prime$, $v^\prime$ is a cut vertex in $\mathcal{G}(D)$,
contradicting to the fact that $\mathcal{G}(D)$ is biconnected
(since $D$ is strongly biconnected).

Consequently, when $V(D) \backslash V(D_i)$ is nonempty, there are
always $D_i$-paths in $D$.

We conclude that the claim holds and complete the proof of the
theorem.

\end{proof}

%
%

\noindent We can prove the following structural theorem for strongly
connected digraphs.

\begin{thm}\label{thm:strong2}
Let $D=(V,E)$ be a strongly connected digraph. Then the strongly
biconnected components of $D$ are those $D[V(B)]$, namely the
sub-digraph of $D$ induced by $V(B)$, where $B$ is a biconnected
component of $\mathcal{G}(D)$.
\end{thm}

%
\begin{proof}

Let $D=(V,E)$ be a strongly connected digraph.

If $D$ is trivial, then the result is obvious.

Otherwise, $D$ is nontrivial, let $B$ be a biconnected component of
$\mathcal{G}(D)$.

It is sufficient to show that $D[V(B)]$ is strongly connected. If
this holds, then $D[V(B)]$ is strongly biconnected. Because all the
vertices of a strongly biconnected sub-digraph of $D$ are in some
 biconnected component of $\mathcal{G}(D)$ and $B$ is a biconnected component of $\mathcal{G}(D)$,
$D[V(B)]$ is a maximal strongly biconnected sub-digraph of $D$, i.e.
a strongly biconnected component of $D$. Since the union of all
biconnected components of $\mathcal{G}(D)$ is $\mathcal{G}(D)$
itself, the theorem holds.

Now we show that $D[V(B)]$ is strongly connected.

Let $v,w \in V(B)$ such that $v \ne w$. Since $D$ is strongly
connected, there must be a path $P$ from $v$ to $w$ in $D$. Now we
show that $P$ is in $D[V(B)]$ as a matter of fact.

To the contrary, suppose that there is a vertex on $P$ not in
$D[V(B)]$.

Let $v^\prime$ be the first vertex on $P$ (starting from $v$) not in
$D[V(B)]$. Then there is $w^\prime \in V(B)$ on $P$ such that
$(w^\prime,v^\prime) \in E$. Because $B$ is a biconnected component
of $\mathcal{G}(D)$, and two distinct biconnected components contain
at most one vertex in common according to
Theorem~\ref{thm:biconn-components}, it follows that $v^\prime$ is
in a biconnected component $B^\prime \neq B$ of $\mathcal{G}(D)$,
and $w^\prime$ is the unique vertex shared by $B^\prime$ and $B$.
Since $P$ is a path, we have that $w^\prime \ne w$, otherwise we
have reached $w$ before $v^\prime$ on $P$, a contradiction. Because
$w \in V(B)$ and $w \ne w^\prime$, we have that $w \in V(B)
\backslash V(B^\prime)$. Since $w^\prime$ is the unique vertex
shared by $B$ and $B^\prime$, any path from $v^\prime$ to $w \in
V(B) \backslash V(B^\prime)$ has to visit $w^\prime$, so $P$ must
visit $w^\prime$ again after visiting $v^\prime$, contradicting to
the fact that $P$ is a path and there should be no vertices visited
twice on a path.

Consequently for any $v,w \in V(B)$, $v \ne w$, there is a path in
$D[V(B)]$ from $v$ to $w$, $D[V(B)]$ is strongly connected.
\end{proof}

From Theorem~\ref{thm:strong2}, we have the following definition for
strongly-biconnected-component graph of a strongly connected
digraph.

\begin{defn}
Let $D$ be a strongly connected digraph, the
strongly-biconnected-component graph of $D$, denoted
$\mathcal{G}_{sbc}(D)=(V_{sbc},W_{sbc},E_{sbc})$, is
$\mathcal{G}_{bc}\left(\mathcal{G}(D)\right)$, namely the
biconnected-component graph of the underlying graph of $D$.
\end{defn}

From the above definition and Theorem~\ref{thm:biconn-components},
we have the following corollary.
\begin{cor}\label{cor:G_bis}
Let $D$ be a strongly connected digraph. Then $\mathcal{G}_{sbc}(D)$
is a tree.
\end{cor}

\begin{exmp}[Strongly-biconnected-component graph]
A strongly connected digraph $D$
(Figure~\ref{fig:strong-example}(a)) and its
strongly-biconnected-component graph $\mathcal{G}_{sbc}(D)$
(Figure~\ref{fig:strong-example}(b)).

\begin{figure}[ht]
\centering
  \includegraphics[width=0.75\textwidth]{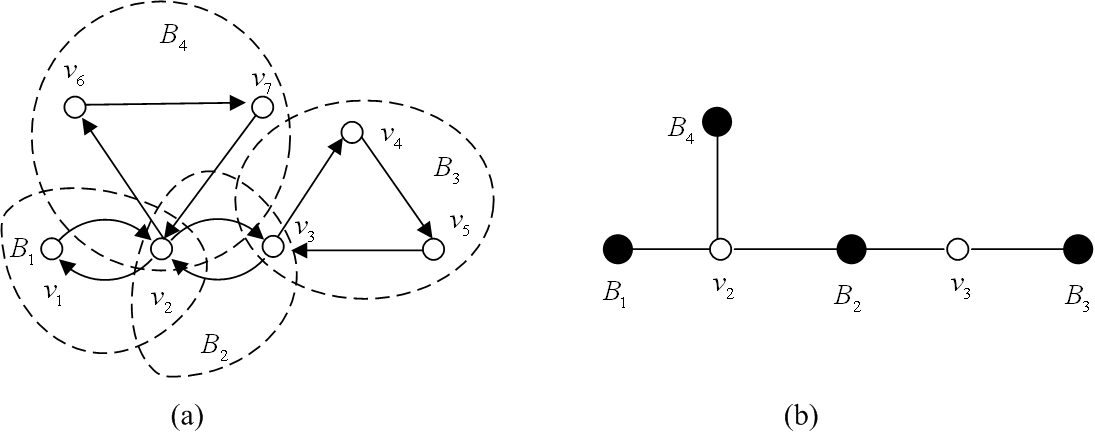}
  \caption{Example: strongly-biconnected-component graph}\label{fig:strong-example}
\end{figure}
\end{exmp}

\section{Motion planning on strongly connected digraphs}\label{sec:strong-algo}

At first, we make the following observation about motion planning on
strongly connected digraphs.

\begin{prop}\label{prop:motion-strong-observ}
Let $D=(V,E)$ be a strongly connected digraph. Then
\begin{enumerate}
\item If the robot and a hole are in the same cycle $C$ of $D$, then the robot can be moved to any vertex of $C$.

\item The movement of objects (robot or obstacles) in $D$ preserves the feasibility of motion
planning on $D$.
\end{enumerate}
\end{prop}
%
\begin{proof}

\noindent (i): it is obvious since the hole can be moved along the
reverse direction of the arcs in $C$ and the objects can be rotated
to any vertex in $C$.

\noindent (ii): Suppose we move an object from $v$ to $w$ along the
arc $(v,w) \in E$. We prove that the motion planning problem is
feasible before the movement iff it is feasible after the movement.

Since $(v,w) \in E$ and $D$ is strongly connected, there is a path
$P$ from $w$ to $v$ in $D$, let $C$ denote the cycle $P \cup
\{(v,w)\}$.

Suppose the motion planning problem is feasible before the movement.
Because after the movement, there is a hole in $v$, we can move the
hole along the reverse direction of $C$, rotate the objects along
$C$, and restore the situation before the movement, namely all the
objects return to the positions before the movement. An example of
this restoration is given in Figure~\ref{fig:cycle-rotate}. So the
motion planning problem is also feasible after the movement.

\begin{figure}[ht]
\centering
  \includegraphics[width=0.6\textwidth]{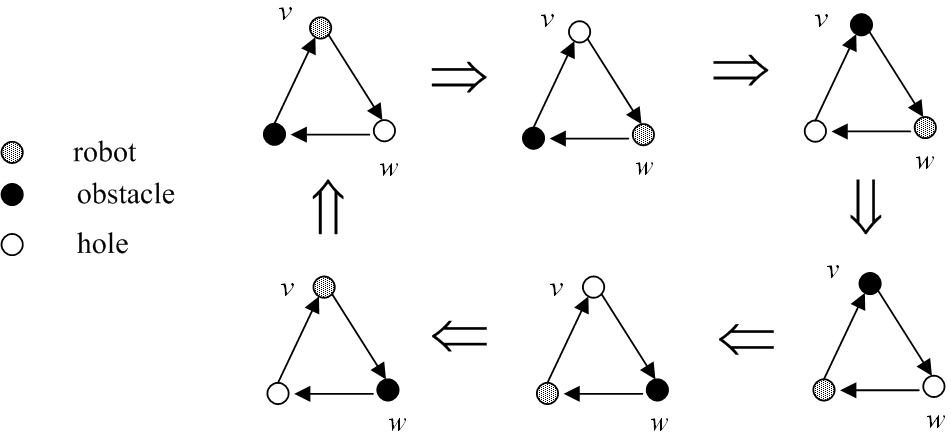}
  \caption{Restoration by rotating the objects in a cycle}\label{fig:cycle-rotate}
\end{figure}

The other direction is obvious.
\end{proof}

%
%
\smallskip
From \cite{PapadimitriouRST94}, we know that if a graph is
biconnected, then one hole is sufficient to move the robot from the
source vertex to the destination vertex, which is also the case for
strongly biconnected digraphs.
%
%
\begin{thm}\label{thm:feas-bi-strong}
Let $D$ be a strongly biconnected digraph. Then the motion planning
problem on $D$ is feasible iff there is at least one hole in $D$.
\end{thm}
%
\begin{proof}

``Only if'' part: obvious.

``If'' part:

Suppose $D$ is strongly biconnected, there is exactly one hole in
$D$ (the case that there are more than one hole is similar), the
source vertex is $s$ and the destination vertex is $t$.

From Theorem~\ref{thm:bi-strong}, we know that there is an open ear
decomposition $P_0,...,P_r$ of $D$.

Let $j_0$ be the minimal $j$ such that $s, t$ and the hole are all
in $D_j$, where $D_j=\bigcup \limits_{0 \le j^\prime \le j}
P_{j^\prime}$.

Induction on $j_0$.

Induction base $j_0=0$: $s$, $t$ and the hole are all in the cycle
$P_0$. Then move the hole along the reverse direction of the cycle
and move the robot to $t$.

Induction step $j_0 > 0$.

Let the tail and head endpoint of $P_{j_0}$ be $u^\prime$ and
$v^\prime$ respectively.

Because of minimality of $j_0$, we have the following three cases.

\textbf{Case I} $s$ is in $P_{j_0}$, $s \neq u^\prime, v^\prime$:

Select a path $P$ in $D_{j_0-1}$ from $v^\prime$ to $u^\prime$, then
$P_{j_0} \cup P$ is a cycle in $D$.

If the hole is not in $P_{j_0} \cup P$, the hole must be in
$D_{j_0-1}$, we can move it to $P$ in $D_{j_0}$ without moving the
robot in $s$.

If $t$ is in $P_{j_0} \cup P$, then move the hole along the reverse
direction of $P_{j_0} \cup P$ and move the robot to $t$.

Otherwise, move the hole along the reverse direction of $P_{j_0}
\cup P$ and move the robot to $v^\prime$. Now the hole is in
$P_{j_0}$, move the hole along the reverse direction of $P_{j_0}$,
until it reaches $u^\prime$.

Then the position of the robot, $v^\prime$, the destination $t$ and
the position of the hole, $u^\prime$, are all in $D_{j_0-1}$,
according to the induction hypothesis, we can move the robot to $t$.

\textbf{Case II} $s$ is in $D_{j_0-1}$, the hole is in some vertex
of $P_{j_0}$ different from $u^\prime$ and $v^\prime$:

Select a path $P$ in $D_{j_0-1}$ from $v^\prime$ to $u^\prime$, then
$P_{j_0} \cup P$ is a cycle in $D$.

If $t$ is in $D_{j_0-1}$ and $s \neq u^\prime$, we can move the hole
to $u^\prime$ along the reverse direction of $P_{j_0}$ without
moving the robot in $s$, then according to the induction hypothesis,
we can move the robot to $t$.

If $t$ is in $D_{j_0-1}$ and $s = u^\prime$, then move the hole
along the reverse direction of $P_{j_0} \cup P$ and move the robot
to $v^\prime$. Now the hole is in $P_{j_0}$, we can move the hole
along the reverse direction of $P_{j_0}$ to $u^\prime$. Then by the
induction hypothesis, we can move the robot to $t$.

If $t$ is not in $D_{j_0-1}$, then $t$ is in $P_{j_0}$.

If $s = u^\prime$, then we can move the hole along the reverse
direction of $P_{j_0} \cup P$ and move the robot to $t$.

Now we consider the case $s \neq u^\prime$.

We can move the hole along the reverse direction of $P_{j_0}$ to
$u^\prime$ without moving the robot. Then by the induction
hypothesis, we can move the robot from $s$ to $v^\prime$ in
$D_{j_0-1}$.

By the induction hypothesis again, we can move the robot from
$v^\prime$ to $u^\prime$ in $D_{j_0-1}$. Let the trace of the robot
during the movement from $v^\prime$ to $u^\prime$ be $P^\prime$.
Note that $P^\prime$ may contain cycles. Suppose the last arc of
$P^\prime$ is $(w,u^\prime)$ for some $w$. Then the hole is in $w$
after the movement. Without loss of generality, we assume that
during the movement, the robot visits $u^\prime$ only once since
$u^\prime$ is the destination. Consequently, the hole can be moved
from $w$ to $v^\prime$ along the reverse direction of $P^\prime$
without moving the robot in $u^\prime$ (see
Figure~\ref{fig:strong-biconn-motion}).

\begin{figure}[ht]
\centering
  \includegraphics[width=0.4\textwidth]{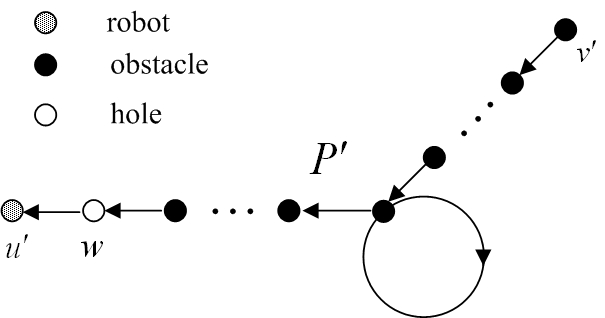}
  \caption{The case that the hole and $t$ are in $P_{j_0}$ different from $u^\prime$ and $v^\prime$, $s \in D_{j_0-1}$, and $s \neq u^\prime$}\label{fig:strong-biconn-motion}
\end{figure}

Since $P_{j_0} \cup P$ is a cycle, now we can move the hole along
the reverse direction of $P_{j_0} \cup P$ and move the robot to $t$.

\textbf{Case III} $s$ and the hole are both in $D_{j_0-1}$, $t$ is
in $P_{j_0}$, $t \neq u^\prime, v^\prime$:

We can move the hole in $D_{j_0-1}$ to $v^\prime$ with possible
movements of the robot in $D_{j_0-1}$. Suppose the new position of
the robot is $s^\prime$.

Now move the hole to some vertex in $P_{j_0}$ different from
$u^\prime$ and $v^\prime$, which is possible since $P_{j_0}$
contains at least three vertices. Then we have reduced Case III to
Case II.
\end{proof}
%
%
\noindent We introduce the following notation before giving the
algorithm.
\begin{defn}\label{defn:w-side-v}
Let $D=(V,E)$ be a strongly connected digraph, $u, v, w \in V$ such
that $v \ne w$, and $\mathcal{G}_{sbc}(D)=(V_{sbc}, W_{sbc},
E_{sbc})$ be the strongly-biconnected-component graph of $D$. Then
$u$ is said to be on the $w$-side of $v$, if $u \ne v$ and one of
the following two conditions holds:
\begin{enumerate}
\item $v \in W_{sbc}$, and $u, w$ are in the same connected component of $\mathcal{G}(D)-v$.

\item $v \not \in W_{sbc}$, and either $u, w$ are in the same connected component of $\mathcal{G}(D-V(B))$, or
$u \in V(B)$, where $B$ is the unique strongly biconnected component
of $D$ to which $v$ belongs.
\end{enumerate}
$u$ is said to be on the non-$w$-side of $v$ if $u \ne v$, and $u$
is not on the $w$-side of $v$.

A hole (resp. obstacle) is said to be on the $t$-side of the robot
if the position (vertex) of the hole (resp. obstacle) is on the
$t$-side of the position of the robot, and a hole (resp. obstacle)
is said to be on the non-$t$-side of the robot if the position of
the hole is on the non-$t$-side of the position of the robot.
\end{defn}

Note that if $u,v,w \in V$, $v \not \in W_{sbc}$, $v \ne w$, $v, w
\in V(B)$, where $B$ is the unique strongly biconnected component of
$D$ to which $v$ belongs, then $u$ is on the $w$-side of $v$ iff $u
\ne v$ and $u \in V(B)$ according to Definition~\ref{defn:w-side-v}.

\begin{exmp}[$t$-side of the robot]
In Figure~\ref{fig:w-side-v-example}(a), the robot is in $s \in
W_{sbc}$, two holes in $v_3$ and $v_4$ are on the $t$-side of the
robot, and the hole in $v_1$ is on the non-$t$-side of the robot. In
Figure~\ref{fig:w-side-v-example}(b), the robot is in $v_3 \not \in
W_{sbc}$, the hole in $v_2$ belongs to the same strongly biconnected
component as $v_3$, so $v_2$ is on the $t$-side of the robot, and
two holes in $v_1$ and $s$ are on the non-$t$-side of the robot.
\end{exmp}

\begin{figure}[ht]
\centering
  \includegraphics[width=0.7\textwidth]{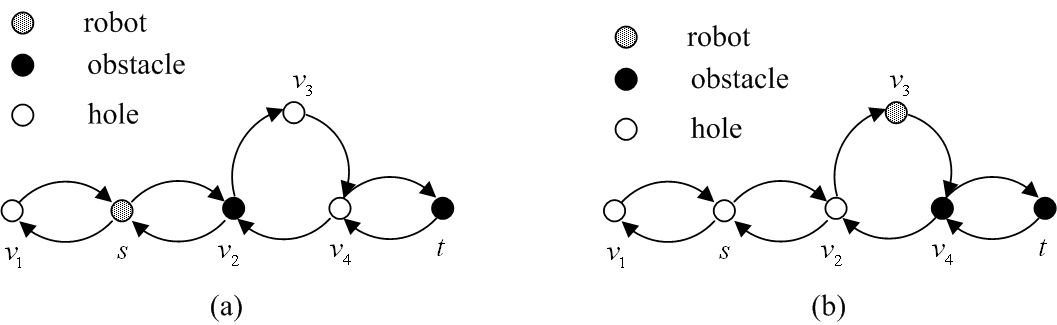}
  \caption{Example: $t$-side of the robot}\label{fig:w-side-v-example}
\end{figure}

%

Function FSCD (see Algorithm~\ref{Alg:feas-strong}) decides the
feasibility of motion planning problem on strongly connected
digraphs. FSCD is similar to the algorithm for motion planning on
graphs since strongly biconnected components of strongly connected
digraphs are similar to biconnected components of connected graphs.
%
%
\begin{algorithm}[ht]

\dontprintsemicolon \SetVline

\SetKwInOut{Input}{input}
\SetKwInOut{Output}{output}
\Input{$(D,s,t,f)$ such that $D=(V,E)$ is a strongly connected digraph, $s,t \in V$, $s \ne t$, and $f$ is a function from $V$ to \{``robot'',``obstacle'',``hole''\}.}
\Output{true or false.}


Construct the underlying graph of $D$, $\mathcal{G}(D)$, and
construct the biconnected-component graph of $\mathcal{G}(D)$ to get
$\mathcal{G}_{sbc}(D)=(V_{sbc},W_{sbc},E_{sbc})$.\;

\While{there are obstacles on the $t$-side of the robot and there
are holes on the non-$t$-side of the robot}
{
  Let $v \in V$ be the current position of the robot.\;
  \eIf{$v \in W_{sbc}$}
  {
    Select a strongly biconnected component $B$ such that $v \in V(B)$, all the vertices of $B$ are not on the $t$-side of $v$,
    and there is at least one hole on the $w$-side of $v$ for some $w \in V(B)$, $w \ne v$.\;

    \If{there are no holes in $B$}
    {
      There is $w \in V(B)$ and $w \in W_{sbc}$ such that there is at least one hole on the non-$t$-side of $w$,
      move one such hole to $w$ without moving the robot.
    }

    Move a hole in $B$ to $v$ and the robot is moved to some vertex $w^\prime \in V(B)$ such that $(v,w^\prime) \in E$.\;
  }
  {
%
    Let $B$ be the unique strongly biconnected component to which $v$ belongs ($B$ contains at least three vertices).\;

    \If{there are no obstacles in $B$}
    {
      Move an obstacle on the $t$-side of the robot into $B$ without moving the robot.
    }
    Move a hole on the non-$t$-side of the robot into $B$ by moving the robot if necessary, while keeping the robot inside $B$.\;

    Move the robot to $v$ again.
  }
}
Let the current position of the robot be $s^\prime$.\;

\eIf{$s^\prime$ and $t$ are in the same strongly biconnected
component}
{
  \eIf{there is at least one hole on the $t$-side of the robot}
  {\Return true.}
  {\Return false.}\;
}
{
  Let $P=B_0v_1B_1...B_{r-1}v_rB_r$ be the path in
  $\mathcal{G}_{sbc}(D)=(V_{sbc},W_{sbc},E_{sbc})$, such that $s^\prime \in
  B_0$, $t \in B_r$, $s^\prime \neq v_1$ and $t \neq v_r$.\;
  Let $l$ be the maximum of $j-i+1$ such that $1 \le i \le j \le r$, and
  $i,j$ satisfy the following conditions:\;
  \Indp 1. $i=1$, or $B_{i-1}$ contains at least three vertices, or
  there is some $B \in V_{sbc}$ such that $B$ is not on $P$ and $\{B,v_i\} \in E_{sbc}$;\;
  2. $j=r$, or $B_j$ contains at least three vertices, or there is
  some $B \in V_{sbc}$ such that $B$ is not on $P$ and $\{B,v_j\} \in E_{sbc}$;\;
  3. For all $i \le k < j$, $B_{k}$ contains only two vertices, and
  for all $i < k < j$, there is no $B \in V_{sbc}$ such that $B$ is
  not on $P$ and $\{B,v_k\} \in E_{sbc}$.\;
  \Indm \eIf{the number of holes on the $t$-side of the robot is no
  less than $l+1$}
  {\Return true.}
  {\Return false.}\;
}
\caption{FSCD($D,s,t,f$)\label{Alg:feas-strong}}
\end{algorithm}
\clearpage

\begin{exmp}[Computation of FSCD]
The strongly connected digraph is given in
Figure~\ref{fig:feas-strong-example-1}(a). At first, the robot is
moved from $s$ to $v_1$, all the holes are on the $t$-side of the
robot (see Figure~\ref{fig:feas-strong-example-1}(b)). Then
according to the definition of $l$ in FSCD, we have $l=3$. There are
four holes on the $t$-side of the robot, so FSCD returns ``true''.
Now we show how the robot is moved from $v_1$ to $t$ with the four
holes: three holes are moved to $s,v_2,v_8$ and the robot is moved
to $v_8$ (see Figure~\ref{fig:feas-strong-example-1}(c)), then the
holes are moved to $v_2,v_3,v_4,v_9$ (See
Figure~\ref{fig:feas-strong-example-1}(d)), the robot is moved to
$v_9$ (see Figure~\ref{fig:feas-strong-example-1}(e)), and all the
holes are moved to $v_5,v_6,v_7,t$ (see
Figure~\ref{fig:feas-strong-example-1}(f)), finally the robot is
moved to $t$.
\end{exmp}

\begin{figure}[ht]
\centering
  \includegraphics[width=0.75\textwidth]{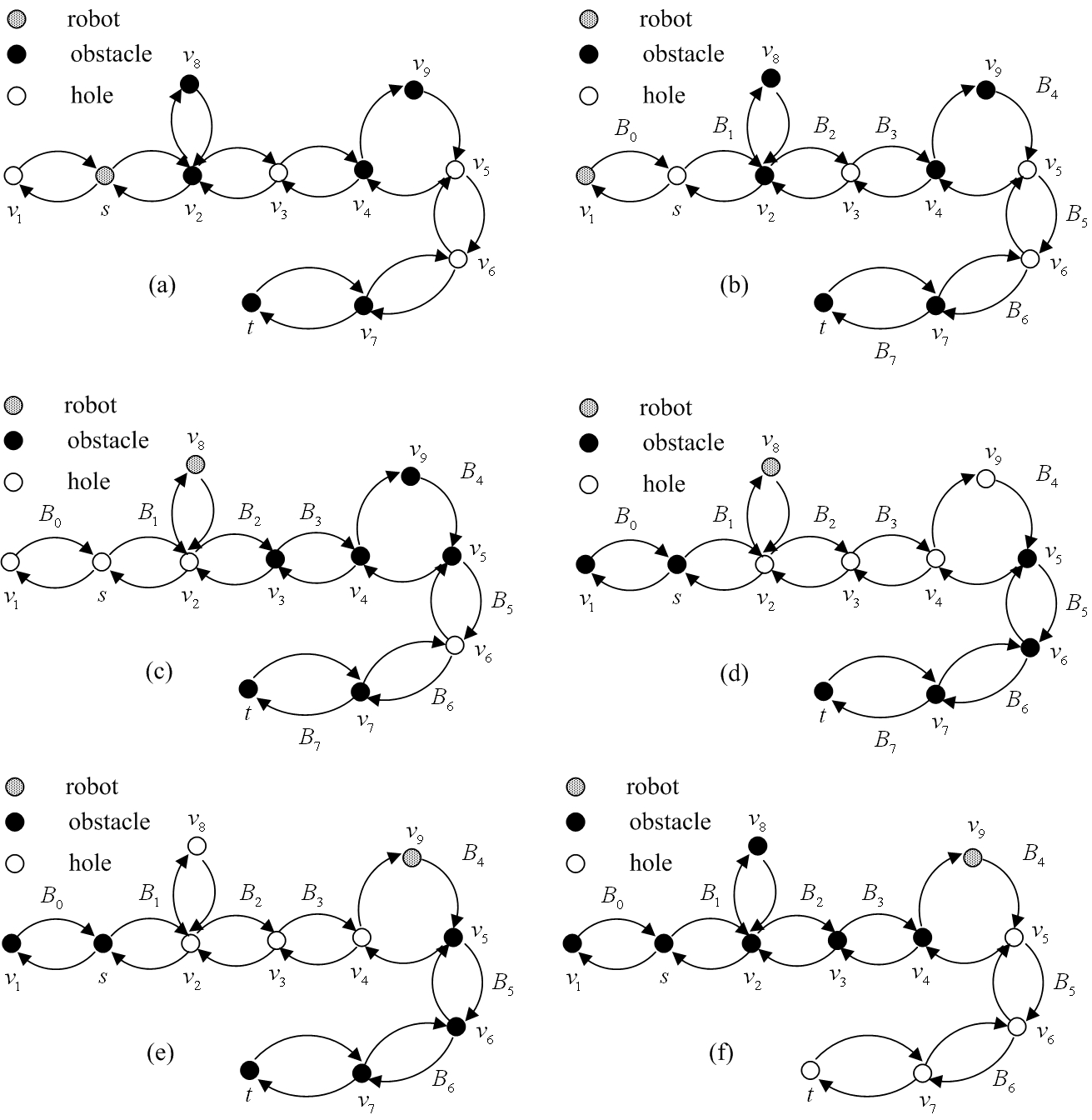}
  \caption{Example: motion planning on a strongly connected digraph}\label{fig:feas-strong-example-1}
\end{figure}

\begin{thm}
FSCD is correct.
\end{thm}
%
\begin{proof}

At first, we show that the ``While'' loop in FSCD terminates.

It is sufficient to show that each execution of the body of the
``While'' loop reduces the number of holes on the non-$t$-side of
the robot by $1$.

There are two cases.

\textbf{Case the robot is in some $v \in W_{sbc}$}.

Then $v$ is shared by several strongly biconnected components.

Because there are holes on the non-$t$-side of the robot, we can
select a strongly biconnected component $B$ such that $v \in V(B)$,
all the other vertices of $B$ different from $v$ are on the
non-$t$-side of $v$, and there is at least one hole on the $w$-side
of $v$ for some $w \in V(B)$, $w \ne v$.

If there are no holes in $B$, then there is $w^\prime \in V(B)$ and
$w^\prime \in W_{sbc}$ such that there is at least one hole on the
non-$t$-side of $w^\prime$, then one such hole can be moved to
$w^\prime$ without moving the robot.

Now there must be at least one hole in some $w^\prime \in V(B)$,
there is a path from $v$ to $w^\prime$ in $B$, we can move the hole
from $w^\prime$ to $v$ along the reverse direction of the path, and
the robot is moved to $w^{\prime\prime}$ on the path such that
$(v,w^{\prime\prime}) \in E$.

The hole moved to $v$ is on the non-$t$-side of the robot before the
movement. Now we show that the hole (in $v$) is on the $t$-side of
the robot after the movement according to
Definition~\ref{defn:w-side-v}: if $w^{\prime\prime} \in W_{sbc}$,
then $v$ is in the same connected component as $t$ in
$\mathcal{G}(D)-w^{\prime\prime}$, the hole in $v$ is on the
$t$-side of the robot ($w^{\prime\prime}$); if $w^{\prime\prime}
\not \in W_{sbc}$, since $B$ is the unique strongly biconnected
component to which $w^{\prime\prime}$ belongs, and $v$ is in $B$, so
the hole in $v$ is on the $t$-side of the robot ($w^\prime$) as
well.

Consequently, in the case that the robot is in some $v \in W_{sbc}$,
each execution of the ``While''-loop reduces the number of holes on
the non-$t$-side of the robot by $1$.

\textbf{Case the robot is in some $v \not \in W_{sbc}$}.

Let $B$ be the unique strongly biconnected component to which $v$
belongs.

Since there are holes on the non-$t$-side of the robot, and
according to Definition~\ref{defn:w-side-v}, holes in $B$ are on the
$t$-side of the robot, there must be some $w \in V(B)$, $w \ne v$,
$w \in W_{sbc}$ such that there is at least one hole on the
non-$t$-side of $w$.

Because there are obstacles on the $t$-side of the robot, if there
are no obstacles in $B$, we can move an obstacle on the $t$-side of
$v$ into $B$ without moving the robot. Now there must be at least
one obstacle in $B$.

If $w$ is not occupied by an obstacle, then an obstacle in $B$ can
be moved to $w$ by moving the robot if necessary. Now a hole on the
non-$t$-side of $w$ can be moved to $w$. Move the robot to $v$
again.

In this case, one hole on the non-$t$-side of the robot is moved
into $B$ and the robot returns to $v$ after the movement.
Consequently, in this case, the number of holes on the non-$t$-side
of the robot is reduced by $1$ as well.

\medskip

\noindent After the execution of the ``While'' loop, either there
are no obstacles on the $t$-side of the robot or all the holes are
on the $t$-side of the robot. In the former case, it is evident that
FSCD returns ``true'' eventually. Now we consider the latter case.

Suppose the current position of the robot is $s^\prime$ now.

If $s^\prime$ and $t$ are in the same strongly biconnected component
$B$, then it is easy to see that the problem is feasible iff there
is at least one hole on the $t$-side of the robot according to
Theorem~\ref{thm:feas-bi-strong}.

Otherwise, let $l$ be the number as defined in FSCD, we show that
the problem is feasible iff there are at least $l+1$ holes.

``Only If'' part: Suppose the problem is feasible.

Then according to Proposition~\ref{prop:motion-strong-observ}, it is
still feasible after the execution of the ``While''-loop.

To the contrary, suppose that there are at most $l$ holes.

Let $P=B_0v_1B_1...B_{r-1}v_rB_r$ be the path in
$\mathcal{G}_{sbc}(D)=(V_{sbc},W_{sbc},E_{sbc})$, such that
$s^\prime \in B_0$, $t \in B_r$, $s^\prime \neq v_1$ and $t \neq
v_r$.

Let $i,j: 1 \le i,j \le r$ satisfy Condition 1-3 in FSCD and
$l=j-i+1$.

Since the problem is feasible, during the movement of the robot from
$s^\prime$ to $t$, the robot should be moved to $v_i$ sometime.

If the robot has been moved to $v_i$, then there must be one hole on
the non-$t$-side of $v_i$. So there are at most $l-1$ holes on the
$t$-side of $v_i$. Since $l-1$ holes are needed to occupy all the
vertices $v_{i+1},\cdots,v_j$ and move the robot from $v_i$ to
$v_j$, if the robot has been moved from $v_i$ to $v_j$, then all the
holes are on the non-$t$-side of $v_j$ now. The robot cannot be
moved further towards $t$, namely the robot cannot be moved to the
vertices on the $t$-side of $v_j$, the problem is infeasible, a
contradiction.

``If'' part: Suppose there are at least $l+1$ holes on the $t$ side
of the robot.

Now we show how to move the robot from $s^\prime$ to $t$.

Let $i_1\cdots i_p$ ($i_1 < i_2 < ... < i_p$) be the list of all the
numbers $i_j$ such that $1 \le i_j < r$, and one of the following
two conditions holds,
\begin{itemize}
\item $B_{i_j}$ contains at least three vertices,

\item there is some $B \in V_{sbc}$ not on $P$ satisfying that $\{B, v_{i_j}\} \in
E_{sbc}$.
\end{itemize}

Without loss of generality, assume that there is at least one $i_j$
satisfying the above condition. The case that there are no such
$i_j$'s can be discussed similarly.

By convention, let $i_0=0$.

At first, we show how to move the robot from $B_{i_0}$ to $B_{i_1}$
if $i_1$ satisfies the first condition, and how to move the robot
from $B_{i_0}$ to some $v$ in $B$ such that $B$ is not on $P$,
$\{B,v_{i_1}\} \in E_{sbc}$, and $(v_{i_1},v) \in E$, if $i_1$
satisfies the second condition.

If $B_{i_1}$ contains at least three vertices, since $l \ge i_1$ and
all the holes are on the $t$-side of the robot, we can move the
holes to occupy all the $v_j$'s such that $1 \le j \le i_1$ and let
another hole occupy some vertex in $B_{i_1}$ different from
$v_{i_1}$. Then we move the robot to $v_1$ (which is possible
according to Theorem~\ref{thm:feas-bi-strong}). We continue moving
the robot to $v_2$, $\cdots$, until to $v_{i_1}$. Moreover, because
there is still one hole in $B_{i_1}$, we can move the robot inside
$B_{i_1}$ and move one hole to $v_{i_1+1}$ and all the other holes
to the $t$-side of the $v_{i_1+1}$.

If there is some $B$ not on $P$ such that $\{B, v_{i_1}\} \in
E_{sbc}$, since $l \ge i_1$, we let the holes occupy all the $v_j$'s
such that $1 \le j \le i_1$ and another hole occupy some $v \in
V(B)$ such that $(v_{i_1},v) \in E$. Then we move the robot to
$v_{i_1}$, and to $v$. Now we can move one hole to $v_{i_1}$ and all
the other holes to the $t$-side of $v_{i_1}$.

The discussions for $i_j$ and $i_{j+1}$ ($2 \le j < p$) are similar
to the above discussion.

During the movement, if sometime there are no obstacles on the
$t$-side of the robot, then obviously the robot can be moved to $t$
and the problem is feasible. In the following we consider situations
that such situation does not occur.

Now we assume that
\begin{enumerate}
\item If $B_{i_{p}}$ contains at least three vertices, then the robot is in $B_{i_{p}}$, one hole
is in $v_{i_p+1}$, and all the other holes are on the $t$-side of
$v_{i_p+1}$.

\item If there is $B \in V_{sbc}$ not on $P$ such that
$\{B,v_{i_p}\}
\in E_{sbc}$, then the robot is in some $v \in V(B)$ such that
$\left(v_{i_p},v\right) \in E$, one hole is in $v_{i_p}$, and all
the other holes are on the $t$-side of $v_{i_p}$.
\end{enumerate}

In the first case above, since $l \ge r-i_p$, we can move one hole
to some $w \in V(B_r)$ such that $w \ne v_r$, and the other holes to
occupy vertices $v_{i_p+1},\cdots, v_r$. Then we can move the robot
to $v_{i_p+1}$, $\cdots$, until to $v_r$. Finally move the robot to
$t$ inside $B_r$.

In the second case above, since $l \ge r-{i_p}+1$, we can move one
hole to some $w \in V(B_r)$ such that $w \ne v_r$, and the other
holes to occupy $v_{i_p},\cdots, v_r$. Then we can move the robot to
$v_{i_p}$, $\cdots$, until to $v_r$. Finally move the robot to $t$
inside $B_r$.

\end{proof}

\begin{thm}\label{thm:strong-complexity}
The time complexity of FSCD is $O(nm)$, where $n$ is the number of
vertices and $m$ is the number of arcs.
\end{thm}
%
\begin{proof}

There are three phases in FSCD: the phase constructing
$\mathcal{G}_{sbc}(D)$, the phase of the ``While''-loop, and the
phase checking whether the number of holes are sufficient to move
the robot to the destination.

The phase constructing $\mathcal{G}_{sbc}(D)$ is in time $O(m)$
since the biconnected components of a connected graph of $m$ edges
can be constructed in $O(m)$ time by a depth-first-search technique
\cite{CLRS01}.

Each execution of the ``While''-loop takes $O(m)$ time, and there
are at most $n$ such executions since there are at most $n$ holes,
so the ``While'' loop takes $O(nm)$ time in total.

The phase checking whether the number of holes are sufficient to
move the robot to the destination takes $O(m)$ time as well.

So the total time of FSCD is $O(nm)$.
\end{proof}

\section{Conclusion}

In this paper, we considered the feasibility of motion planning on
digraphs, and proposed two algorithms to decide the feasibility of
motion planning on acyclic and strongly connected digraphs
respectively, we proved the correctness of the two algorithms and
analyzed their time complexity. We showed that the feasibility of
motion planning on acyclic and strongly connected digraphs can be
decided in time linear in the product of the number of vertices and
the number of arcs.

The algorithm for the feasibility of motion planning on acyclic
digraphs (FAD) can be adapted to the case where the capacity of each
vertex is more than one (namely, vertices are able to hold several
objects simultaneously), by just changing the computation of the
$h(v)$'s, the number of holes that could be moved to each node $v$.
The algorithm for the feasibility of motion planning on strongly
connected digraphs (FSCD) can also be adapted to the case where the
capacity of each vertex is more than one by only changing the
``While''-loop.

The strongly biconnected digraphs introduced in this paper may be of
independent interest in graph theory since they admit nice
characterization: a nontrivial digraph is strongly biconnected iff
it has an open ear decomposition. It seems interesting to consider
also strongly triconnected digraphs, strongly four-connected
digraphs, etc. and investigate their theoretical properties.

The feasibility of motion planning on digraphs is only partially
solved in this paper since we did not give the algorithm for
deciding the feasibility on general digraphs, which, as well as the
optimization of the motion of robot and obstacles, is much more
intricate than that on graphs because of the irreversibility of the
movements on digraphs.

The motion planning on graphs with one robot, GMP1R, has a natural
generalization, GMP$k$R, where there are $k$ robots with their
respective destinations. It is also interesting to consider motion
planning on digraphs with $k$ robots since in practice it is more
reasonable that a robot shares its workspace with other robots.

GMP$k$R in general is a very complex problem. A special case of
GMP$k$R, where there are no additional obstacles (thus all the
movable objects have their destinations), has been considered.
Wilson studied the special case of GMP$k$R for $k=n-1$ in
\cite{Wilson74}, which is a generalization of the ``15-puzzle''
problem to general graphs. They gave an efficiently checkable
characterization of the solvable instances of the problem.
Kornhauser et al. extended this result to $k \le n-1$
\cite{KornhauserMS84}. Goldreich proved that determining the
shortest move sequence for the problem studied by Kornhauser et al.
is NP-hard \cite{Goldreich84}. It seems more realistic to first
consider the above special case of GMP$k$R on digraphs.
%
%
\bibliographystyle{amsalpha}
\bibliography{biblio-motion}

\end{document}